\documentclass{article}

\usepackage[affil-it]{authblk}
\usepackage[dvipsnames]{xcolor}
\usepackage{amsfonts}
\usepackage{amsmath,amsthm,amssymb,dsfont}
\usepackage{enumerate}
\usepackage{graphicx}
\usepackage{subcaption}
\usepackage[margin=3cm]{geometry}
\usepackage{lmodern}
\usepackage{url}
\usepackage{todonotes}
\usepackage{bbm}

\usepackage{tikz}
\usepackage[a-2b]{pdfx} 

\usepackage{pifont}
\usepackage{multirow}

\usepackage{titlesec}

\titleformat*{\section}{\fontsize{13pt}{16pt}\selectfont \bfseries}
\titleformat*{\subsection}{\fontsize{11pt}{15pt}\selectfont \bfseries}

\usetikzlibrary{shapes.symbols,patterns} 
\usepackage{pgfplots}
\pgfplotsset{compat=1.10}
\usepgfplotslibrary{fillbetween}

\definecolor{linkblue}{HTML}{001487}
\usepackage{hyperref}
\hypersetup{colorlinks=true,citecolor=linkblue,linkcolor=linkblue,filecolor=linkblue,urlcolor=linkblue,breaklinks=true}

\usepackage{mathtools}

\usepackage{zref-clever}
\zcsetup{nameinlink = tsingle, cap, abbrev}
\newcommand{\cref}[1]{\zcref{#1}}
\newcommand{\Cref}[1]{\zcref[S]{#1}} 

\usepackage{mdframed}
\usepackage{aligned-overset}

\theoremstyle{plain}
\newtheorem{theorem}{Theorem}[section]

\newtheorem{lemma}{Lemma}[section]

\theoremstyle{definition}

\newtheorem{remark}[theorem]{Remark}

 \usepackage{tcolorbox}
\tcolorboxenvironment{definition}{}
\tcolorboxenvironment{theorem}{colback=RoyalBlue!5!white,colframe=RoyalBlue!75!black}
\tcolorboxenvironment{lemma}{colback=yellow!5!white,colframe=yellow!75!black}
\tcolorboxenvironment{corollary}{colback=Dandelion!5!white,colframe=Dandelion!75!black}
\tcolorboxenvironment{proposition}{colback=Emerald!5!white,colframe=Emerald!75!black}
\tcolorboxenvironment{claim}{colback=pink!25!white,colframe=pink!100!black}
\tcolorboxenvironment{conjecture}{colback=red!5!white,colframe=red!75!black}
\tcolorboxenvironment{example}{colback=lightgray!5!white,colframe=lightgray!75!black}
\newtcolorbox{theorem-box}{colback=red!5!white,colframe=red!75!black}

\newcommand*{\cF}{\mathcal{F}}

\newcommand*{\cP}{\mathcal{P}}

\newcommand*{\cS}{\mathcal{S}}

\newcommand*{\id}{\mathds{1}}
\newcommand*{\poly}{\mathrm{poly}}

\newcommand*{\tr}{\mathrm{tr}}

 \allowdisplaybreaks

\title{The marginal is pretty good}

 \author{Lukas Schmitt$^{1,2}$, Gui-Long Jiang$^{3,4}$, Shi-Bing Li$^{3}$, and Joseph M. Renes$^1$}
\affil{\small $^{1}$Institute for Theoretical Physics, ETH Zurich\\
  $^{2}$IBM Quantum, IBM Research Europe -- Zurich \\
  $^{3}$Institute for Advanced Study in Mathematics, Harbin Institute of Technology, Harbin 150001, China\\
  $^{4}$School of Mathematics, Harbin Institute of Technology, Harbin 150001, China
 }
 \date{}

\begin{document}
\maketitle
\vspace{-14mm}
\begin{abstract}
One-shot information theory measures often require an optimization over states, but the form of these optimizers can be complicated or depend on the initial problem in nonlinear ways. In this note, we show that in many instances using the marginal instead of the optimal state is sufficiently good and only changes the result by a small factor. We prove that for the Petz--R\'enyi and sandwiched--R\'enyi divergences of order $\alpha\in(0,1)$, replacing the optimizing state on $B$ by the marginal $\rho_B$ results in a multiplicative overhead of at most $1/\alpha$.
\end{abstract}

\section{Introduction}
Variational expressions of the form
\begin{equation}\label{eq:variational-problem}
\min_{\sigma_B}D_\alpha(\rho_{AB}\|\tau_A\otimes\sigma_B)
\end{equation}
occur naturally in R\'enyi generalizations of conditional entropies and related one-shot information measures~\cite{Tomamichel2016}. 
Although the optimization gives the best value, its optimizer generally depends on the joint state in a nonlinear way. 
In contrast, the marginal $\rho_B$, is canonical and explicit. 
In this note we want to quantify the mistake made by choosing the marginal instead of the optimizer.

Our first main result shows that for the Petz--R\'enyi divergence and $\alpha\in(0,1)$, the marginal achieves the optimum up to the factor $1/\alpha$:
\begin{equation}
\label{eq:mainresult}
 D_\alpha(\rho_{AB}\| \tau_A\otimes\rho_B) 
\leq \frac{1}{\alpha} \min_{\sigma_B}
D_\alpha(\rho_{AB}\| \tau_A\otimes\sigma_B).
\end{equation}
This loss is independent of dimension and tends to vanish as $\alpha\to1$. We prove this by a Schatten-norm variational formula for the optimization over $\sigma_B$, and an operator inequality between a weighted marginal and $\rho_B^\alpha$ obtained from data processing.

Our second main result shows that \cref{eq:mainresult} holds for the sandwiched divergence with order $\alpha\in(0,1)$. That is,
\begin{equation}
\label{eq:mainresult2}
 \tilde{D}_\alpha(\rho_{AB}\| \tau_A\otimes\rho_B) 
\leq \frac{1}{\alpha} \min_{\sigma_B}
 \tilde{D}_\alpha(\rho_{AB}\| \tau_A\otimes\sigma_B).
\end{equation}
We briefly outline the proof of Eq.~\eqref{eq:mainresult2}. Fix a candidate state $\sigma_B$ and consider the corresponding $n$-copy problem. We first pinch $\sigma_B^{\otimes n}$ with respect to the spectral decomposition of $\rho_B^{\otimes n}$. We then apply a joint pinching associated with the resulting operator, $\rho_B^{\otimes n}$, and $\tau_A^{\otimes n}$. 
This construction reduces the problem to a commuting one while preserving the $B^n$-marginal. 
By permutation symmetry, the number of pinching blocks grows only polynomially with $n$, so the resulting $O(\log n)$ penalty vanishes after normalization by $n$. 
For $\alpha\in[1/2,1)$, the data-processing inequality for the sandwiched R{\'e}nyi divergence allows us to apply Eq.~\eqref{eq:mainresult} to the resulting commuting problem. 
For $\alpha\in(0,1/2)$, where the usual data-processing inequality is unavailable, we introduce a suitably normalized pinched reference state and establish an approximate data-processing inequality whose loss is again of order $O(\log n)$. 
Applying the same asymptotic reduction and then letting $n\to\infty$ proves Eq.~\eqref{eq:mainresult2} for the entire range $\alpha\in(0,1)$.

\section{Preliminaries}\label{sec:preliminaries}

We write $\cS(B)$ for the state space of a quantum system $B$. Powers of positive semidefinite operators are defined on their supports. 
For $\alpha\in(0,1)$ and $\rho,\sigma\in\cS(B)$, the Petz--R\'enyi divergence between $\rho$ and $\sigma$ is given by
\begin{equation}
   D_\alpha(\rho\|\sigma) =\frac{1}{\alpha-1}\log \tr\left[\rho^\alpha \sigma^{1-\alpha}\right]. 
\end{equation}
Meanwhile, $\alpha'=\frac{1-\alpha}{2\alpha}$, the sandwiched--R\'enyi divergence between $\rho$ and $\sigma$ is given by 
\begin{equation}
   \tilde{D}_\alpha(\rho\|\sigma) =\frac{1}{\alpha-1}\log \tr\left[(\sigma^{\alpha'}\rho\sigma^{\alpha'})^\alpha\right].
\end{equation}
If $\rho$ and $\sigma$ commute, the Petz and sandwiched R{\'e}nyi divergences coincide, i.e.,

\begin{equation}\label{eq:commuting-equality}
  D_\alpha(\rho\|\sigma)=\tilde{D}_\alpha(\rho\|\sigma).
\end{equation}

For a Hermitian operator $X=\sum_{i=1}^v x_i\Pi_i$, where the $x_i$ are its
distinct eigenvalues, define the spectral pinching
\begin{equation}\label{eq:pinching-def}
 \cP_X(Y):=\sum_{i=1}^v\Pi_iY\Pi_i,
\end{equation}
and let $\nu(X)$ denote the number of distinct eigenvalues of $X$.

\section{Petz--R\'enyi bound}\label{sec:petz}
We can prove \cref{eq:mainresult} with the following two lemmas.
\begin{lemma}\label{lem:var}
Let $0<\alpha<1$ and $W_B\ge 0$. Then
\begin{equation}
\max_{\sigma_B\in\cS(B)}\tr\left[W_B\sigma_B^{1-\alpha}\right] 
=\tr \left[ W_B^{1/\alpha}\right]^\alpha .
\end{equation}
\end{lemma}

\begin{proof}
For any density operator $\sigma_B$, H\"older's inequality with exponents $p=1/\alpha$ and $q=1/(1-\alpha)$ gives
\begin{align}
\tr\left[W_B\sigma_B^{1-\alpha}\right] \overset{\mathrm{H\ddot older}}&{\leq} \|W_B\|_{1/\alpha} \, \|\sigma_B^{1-\alpha}\|_{1/(1-\alpha)} \\[1mm]
\overset{\mathrm{Def.}}&{=} \tr \left[ W_B^{1/\alpha}\right]^\alpha \tr \left[ \sigma_B\right]^{1-\alpha} \\[1mm]
\overset{\tr\, \sigma_B=1}&{=}
\tr \left[  W_B^{1/\alpha}\right]^\alpha.
\end{align}
Equality is attained by
\begin{equation}
\sigma_B^\star = \frac{W_B^{1/\alpha}}{\tr \, W_B^{1/\alpha}}.
\end{equation}
Therefore
\begin{equation} \label{eq:wb}
\max_{\sigma_B} \tr\left[W_B\sigma_B^{1-\alpha}\right] 
= \tr \left[ W_B^{1/\alpha}\right]^\alpha .
\end{equation}
\end{proof}

\begin{lemma}\label{lem:opbound}
For any $\alpha\in (0,1)$ and $\tau_A$ a density operator, define
\begin{equation}
W_\alpha
=
\tr_A\left[
(\tau_A^{\frac{1-\alpha}{2}}\otimes \id_B)
\rho_{AB}^{\alpha}
(\tau_A^{\frac{1-\alpha}{2}}\otimes \id_B)
\right].
\end{equation}
Then $W_\alpha\leq\rho_B^\alpha$.
\end{lemma}

\begin{proof}
Let $\omega_B$ be a density operator. By data processing of the Petz--Rényi divergence for $\alpha\in(0,1)$~\cite{Petz1986}, we get
\begin{equation}
D_\alpha(\rho_{AB}\| \tau_A\otimes\omega_B)
\ge
D_\alpha(\rho_B\|\omega_B).
\end{equation}
Since $\alpha-1<0$, this is equivalent to
\begin{align}
\tr\left[\rho_{AB}^{\alpha}(\tau_A\otimes\omega_B)^{1-\alpha}\right]
&\overset{\mathrm{DPI}}{\leq}\tr\left[\rho_B^\alpha\omega_B^{1-\alpha}\right].
\end{align}
Then one gets
\begin{align}
\tr\left[\rho_{AB}^{\alpha}(\tau_A\otimes\omega_B)^{1-\alpha}\right]
&= \tr\left[\rho_{AB}^{\alpha}(\tau_A^{1-\alpha}\otimes\omega_B^{1-\alpha})\right] \\[1mm]
\overset{\mathrm{cyclicity}}&{=}
\tr\left[(\tau_A^{\frac{1-\alpha}{2}}\otimes\id_B) \rho_{AB}^{\alpha}(\tau_A^{\frac{1-\alpha}{2}}\otimes\id_B)(\id_A\otimes\omega_B^{1-\alpha})\right] \\[1mm]
\overset{\mathrm{partial\ trace}}&{=} \tr\left[ W_\alpha \omega_B^{1-\alpha} \right].
\end{align}
Hence, for every full-rank density operator $\omega_B$,
\begin{equation}
\tr\left[ W_\alpha \omega_B^{1-\alpha} \right] \leq \tr\left[ \rho_B^\alpha \omega_B^{1-\alpha} \right].
\end{equation}
Because every positive definite operator $Y_B>0$ is proportional to $\omega_B^{1-\alpha}$ for a suitable full-rank state $\omega_B$, this implies
$W_\alpha\leq\rho_B^\alpha$.
\end{proof}

Now we can state and prove our first main result. 
\begin{theorem}\label{thm:petz}
Let $\alpha\in(0,1)$, let $\rho_{AB}$ be a finite-dimensional bipartite state, and let $\tau_A$ be a full-rank density operator. Then
\begin{equation}
D_\alpha(\rho_{AB}\| \tau_A\otimes\rho_B) 
\leq \frac{1}{\alpha} \min_{\sigma_B\in\cS(B)}
D_\alpha(\rho_{AB}\| \tau_A\otimes\sigma_B).
\end{equation}
\end{theorem}

\begin{proof}
Define
\begin{equation}
\label{eq:Wdef}
W_\alpha = \tr_A\left[ (\tau_A^{\frac{1-\alpha}{2}}\otimes \id_B) \rho_{AB}^{\alpha} (\tau_A^{\frac{1-\alpha}{2}}\otimes \id_B) \right].
\end{equation}
Then, for any state $\sigma_B$,
\begin{align}
D_\alpha(\rho_{AB}\| \tau_A\otimes\sigma_B) \overset{\mathrm{Def}}&{=} \frac1{\alpha-1} \log \tr\left[ \rho_{AB}^{\alpha} (\tau_A\otimes\sigma_B)^{1-\alpha} \right] \\
\overset{\mathrm{Def}}&{=}\frac1{\alpha-1} \log \tr\left[ W_\alpha\sigma_B^{1-\alpha} \right].
\end{align}
Hence, by the variational formula,
\begin{align}
\min_{\sigma_B \in \cS(B)} D_\alpha(\rho_{AB}\| \tau_A\otimes\sigma_B)
\overset{\mathrm{\Cref{lem:var}}}&{=}
\frac{\alpha}{\alpha-1} \log \tr W_\alpha^{1/\alpha}.
\end{align}
Therefore $\frac{1}{\alpha} \min_{\sigma_B} D_\alpha(\rho_{AB}\|\tau_A\otimes\sigma_B)
=
\frac1{\alpha-1} \log \tr[ W_\alpha^{1/\alpha}]$. 
It remains to compare $\tr [W_\alpha^{1/\alpha}]$ with
$\tr[W_\alpha\rho_B^{1-\alpha}]$. 

If $\alpha\in[\frac{1}{2},1)$, we have
$0\leq\frac{1-\alpha}{\alpha}\leq1$, and thus $x\mapsto x^{(1-\alpha)/\alpha}$ is operator monotone on $[0,\infty)$. Using~\Cref{lem:opbound},
\begin{equation}\label{eq:op}
W_\alpha^{\frac{1-\alpha}{\alpha}}
\leq
(\rho_B^\alpha)^{\frac{1-\alpha}{\alpha}}
=
\rho_B^{1-\alpha}.
\end{equation}
Consequently,
\begin{align}
\tr \, W_\alpha^{1/\alpha}
&=
\tr\left[
W_\alpha W_\alpha^{\frac{1-\alpha}{\alpha}}
\right]
\\
\overset{\mathrm{\Cref{eq:op}}}&{\le}
\tr\left[
W_\alpha\rho_B^{1-\alpha}
\right].\label{eq:hook}
\end{align}

Next, we prove that Eq.~\eqref{eq:hook} also holds when $\alpha\in(0,\frac{1}{2}).$ Write the spectral decompositions
\begin{equation}
 W_{\alpha}=\sum_i w_i|i\rangle\!\langle i|,\;\;
 \rho_B^\alpha=\sum_j \lambda_j|b_j\rangle\!\langle b_j|.
\end{equation}
\Cref{lem:opbound} implies that for every $i$ we have
\begin{equation}\label{eq:diagonal-order}
 w_i\leq\langle i|\rho_B^\alpha|i\rangle.
\end{equation}
The scalar function $x\mapsto x^{\frac{1-\alpha}{\alpha}}$ is increasing and convex because $\frac{1-\alpha}{\alpha}>1$.
We can therefore apply Jensen's inequality to obtain
\begin{align}
 \langle i| \rho_B^{1-\alpha} |i\rangle
 &=\langle i| (\rho_B^\alpha)^{\frac{1-\alpha}{\alpha}} |i\rangle \\[1mm]
 &=\sum_j |\langle i|b_j\rangle|^2 \lambda_j^{\frac{1-\alpha}{\alpha}}\\
 &\geq \Big(\sum_j |\langle i|b_j\rangle|^2 \lambda_j \Big)^{\frac{1-\alpha}{\alpha}}\\
 &=\langle i|\rho_B^\alpha|i\rangle^{\frac{1-\alpha}{\alpha}}
 \geq w_i^{\frac{1-\alpha}{\alpha}},\label{eq:scalar-jensen-step}
\end{align}
where the last inequality follows from the monotonicity of $x\mapsto x^{(1-\alpha)/\alpha}$.
Multiplying Eq.~\eqref{eq:scalar-jensen-step} by $w_i$ and summing
over $i$ gives
\begin{align}\label{eq:key-trace-lower-range}
\tr\left[W_{\alpha}\rho_B^{1-\alpha}\right]
 &=\sum_i w_i\langle i|\rho_B^{1-\alpha}|i\rangle \\
 &\geq\sum_i w_i^{1+\frac{1-\alpha}{\alpha}}
 =\tr W_{\alpha}^{1/\alpha}.
\end{align}

Since $\alpha-1<0$, applying $\frac1{\alpha-1}\log(\cdot)$ reverses the inequality:
\begin{equation}\label{eq:log}
\frac1{\alpha-1} \log \tr\left[W_\alpha\rho_B^{1-\alpha} \right]
\leq \frac{1}{\alpha-1} \log \tr [W_\alpha^{1/\alpha}].
\end{equation}
Finally,
\begin{align}
D_\alpha(\rho_{AB}\| \tau_A\otimes\rho_B)
\overset{\mathrm{Def}}&{=}
\frac1{\alpha-1} \log \tr\left[ W_\alpha\rho_B^{1-\alpha}\right]
\\
\overset{\mathrm{\Cref{eq:log}}}&{\le}
\frac1{\alpha-1} \log \tr [W_\alpha^{1/\alpha}] \\
\overset{\mathrm{\Cref{lem:var}}}&{=} \frac1\alpha \min_{\sigma_B}
D_\alpha(\rho_{AB}\| \tau_A\otimes\sigma_B).
\end{align}
This proves the claim.
\end{proof}

\begin{remark}
In the limit $\alpha \to 1$, the theorem holds as well and it is known that the optimizer is the marginal.
\end{remark}

\begin{remark}[Non-full-rank states]\label{rem:cont}
The full-rank assumption on $\tau_A$ is only used to avoid support
technicalities. The statement extends to arbitrary density operators
$\tau_A\ge 0$ by a standard continuity argument.
\end{remark}
\begin{proof}
See~\cref{app:1}
\end{proof}

\begin{remark}[Tightness of the Petz--R\'enyi bound]\label{rem:tight}
The bound is tight in certain circumstances, even for classical $\rho_{AB}$. Indeed, consider the distribution $P_{XY}$ with $\mathcal X=\mathcal Y=\{0,\dots,n-1\}$ for some integer $n$ such that $P_{XY}(0,0)=p$ with $p\in (0,1)$ and $P_{X,Y}(x,y)=(1-p)/(n^2-1)$ for $(x,y)\neq (0,0)$. 
Then for any $\alpha\in (0,1)$, $p\in (0,1)$, and $U_X$ the uniform distribution, we have 
\begin{equation}
\lim_{n\to \infty} \frac{1}{\alpha} \min_{R_Y} D_\alpha(P_{XY}\|U_X\times R_Y)-D_\alpha(P_{XY}\|U_X\times P_Y)=0\,. 
\end{equation}
\end{remark}
\begin{proof}
It is enough to show that the inequality \cref{eq:hook} is saturated.
The marginal $P_Y$ is defined by $P_Y(0)=p+(1-p)/(n+1)$ and $P_Y(y)=(1-p)n/(n^2-1)$ for $y\neq 0$. 
For convenience, define $q=(1-p)n/(n^2-1)$.
Meanwhile $W_Y(y)=n^{\alpha-1}\sum_{x\in \mathcal X}P_{XY}(x,y)^\alpha$ from \cref{eq:Wdef}.
Using the definition of $P_{XY}$, this expression is just $W_Y(0)=n^{\alpha-1}(p^\alpha+(n-1)((1-p)/(n^2-1))^\alpha)$ and $W_Y(y)=q^\alpha$ for $y\neq 0$. 
Observe that $W_Y^{1/\alpha}$ and $W_YP_Y^{1-\alpha}$ are both equal to $q$ for  $y\neq 0$. The bound then implies that $W_Y^{1/\alpha}(0)\leq W_Y(0)P_Y(0)^{1-\alpha}$, so to complete the proof it is sufficient to show that $\lim_{n\to \infty} W_Y(0)P_Y(0)^{1-\alpha}=0$. The limit of $P_Y(0)$ is just $p$, while $W_Y(0)$ tends to zero. 
\end{proof}

\section{Sandwiched--R{\'e}nyi bound}\label{sec:fidelity}
In this section, we prove Eq.~\eqref{eq:mainresult2}. We first establish the result for $\alpha\in[\frac{1}{2},1)$ using the Petz--R{\'e}nyi bound. Then, by adapting the same method, we prove the result for $\alpha\in(0,\frac{1}{2})$.

\subsection{Sandwiched--R\'enyi bound with Order \texorpdfstring{$\alpha\in[\frac{1}{2},1)$}{α in [1/2,1)}}
In this subsection, we prove Eq.~\eqref{eq:mainresult2} with order $\alpha\in[\frac{1}{2},1)$. We first recall the following standard bound that controls the loss in the sandwiched R{\'e}nyi divergence induced by pinching.
\begin{lemma}[Lemma~3 of~\cite{HayashiTomamichel2016}]\label{lem:pinching-loss}
Let $\alpha\in(0,1)$, $\rho,\sigma\in\cS(A)$, and let
$\cP(Y)=\sum_{i=1}^v\Pi_iY\Pi_i$ be a pinching satisfying $\cP(\sigma)=\sigma$.  Then
\begin{equation}\label{eq:pinching-loss}
\tilde{D}_\alpha(\rho\|\sigma)
\leq
\tilde{D}_\alpha(\cP(\rho)\|\sigma)+\log v.
\end{equation}
\end{lemma}

We also need a polynomial bound on the number of distinct eigenvalues of a
permutation-invariant operator.

\begin{lemma}[\cite{Hayashi2017}]\label{lem:polynomial-spectrum}
Let $\mathcal{H}\cong\mathbb C^d$, and let $X_n$ be a Hermitian operator on
$\mathcal{H}^{\otimes n}$ that is invariant under conjugation by every permutation of
the tensor factors.  Then
\begin{equation}\label{eq:poly-spectrum}
\nu(X_n)
\leq (n+1)^{d-1}(n+d)^{d(d-1)/2}.
\end{equation}
In particular, $\log\nu(X_n)=O(\log n)$.  For every fixed Hermitian
operator $X$ on $H$,
\begin{equation}\label{eq:tensor-spectrum}
\nu(X^{\otimes n})
\leq\binom{n+d-1}{d-1}
\leq(n+1)^{d-1}.
\end{equation}
\end{lemma}
The following construction plays a crucial role in establishing the result for the sandwiched R{\'e}nyi
divergence when $\alpha\in[\frac{1}{2},1)$.  It is
adapted to one fixed candidate $\sigma_B$ and uses no universal dominating
state.

\begin{lemma}\label{lem:joint-pinching}
Let $\rho_{AB}\in\cS(AB)$,
$\tau_A\in\cS(A)$, $\sigma_B\in\cS(B)$, and fix $n\geq1$.  Define the
pinched comparison state
\begin{equation}\label{eq:pinched-comparator}
 \widehat\sigma_{B^n}:=
 \cP_{\rho_B^{\otimes n}}(\sigma_B^{\otimes n})
\end{equation}
which belongs to $\cS(B^n)$ because spectral pinching is completely
positive and trace preserving, and define the joint pinching
\begin{equation}\label{eq:joint-pinching}
 \mathcal E_n:=
 \cP_{\tau_A^{\otimes n}}^{A^n}
 \circ\cP_{\widehat\sigma_{B^n}}^{B^n}
 \circ\cP_{\rho_B^{\otimes n}}^{B^n}.
\end{equation}
The three spectral pinchings commute.  If $v_n$ is the number of nonzero projections in their common refinement, then
\begin{align}
 v_n
 &\leq \nu(\tau_A^{\otimes n})
 \nu(\rho_B^{\otimes n})
 \nu(\widehat\sigma_{B^n}) \\
 &\leq(n+1)^{d_A+d_B-2}
 \left[1+(n+1)^{d_B-1}
(n+d_B)^{d_B(d_B-1)/2}\right].
 \label{eq:joint-block-bound}
\end{align}
In particular, $\log v_n=O(\log n)$.  Moreover, for
\begin{equation}\label{eq:pinched-rho}
 \widehat\rho_{A^nB^n}:=\mathcal E_n(\rho_{AB}^{\otimes n}),
\end{equation}
we have
\begin{align}
 {\rm Tr}_{A^n}\widehat\rho_{A^nB^n}
 &=\rho_B^{\otimes n},
 \label{eq:marginal-preserved}\\[1mm]
 [\widehat\rho_{A^nB^n},
   \tau_A^{\otimes n}\otimes\rho_B^{\otimes n}]&=0,
 \label{eq:commute-marginal-reference}\\[1mm]
 [\widehat\rho_{A^nB^n},
   \tau_A^{\otimes n}\otimes\widehat\sigma_{B^n}]&=0,
 \label{eq:commute-comparator-reference}\\[1mm]
 \mathcal E_n(\tau_A^{\otimes n}\otimes\rho_B^{\otimes n})
 &=\tau_A^{\otimes n}\otimes\rho_B^{\otimes n},
 \label{eq:fix-marginal-reference}\\[1mm]
 \mathcal E_n(\tau_A^{\otimes n}\otimes\sigma_B^{\otimes n})
 &=\tau_A^{\otimes n}\otimes\widehat\sigma_{B^n},
 \label{eq:map-comparator-reference}
\end{align}
where $[\rho,\sigma]:=\rho\sigma-\sigma\rho$.
\end{lemma}

\begin{proof}
Both $\rho_B^{\otimes n}$ and $\sigma_B^{\otimes n}$ are invariant under conjugation by every permutation of the tensor factors of $B^{\otimes n}$.
Since every spectral projector of $\rho_B^{\otimes n}$ commutes with every permutation operator, the pinched operator
\begin{equation}
\widehat\sigma_{B^n}
=
\cP_{\rho_B^{\otimes n}}(\sigma_B^{\otimes n})
\end{equation}
is permutation invariant as well. Moreover, the defining property of spectral pinching implies that
\begin{equation}
[\widehat\sigma_{B^n},\rho_B^{\otimes n}]=0.
\end{equation}
Hence the spectral projectors of $\widehat\sigma_{B^n}$ commute with those of $\rho_B^{\otimes n}$.
Since the product of two commuting projectors is still a projector and $\mathcal P_{\tau_A^{\otimes n}}$ acts on $\mathcal H^{\otimes n}_A$, whereas $\mathcal P_{\widehat{\sigma}_{B^n}}$ and $\mathcal P_{\rho_B^{\otimes n}}$ act on $\mathcal H^{\otimes n}_B$, $\mathcal E_n$ is a pinching and all three pinchings in Eq.~\eqref{eq:joint-pinching} commute.
Since the projections of the joint pinching are obtained from the nonzero products of the spectral projections of the three commuting pinchings, the number $v_n$ of projections in the pinching $\mathcal E_n$ satisfies
\begin{equation}
v_n
\leq
\nu(\tau_A^{\otimes n})
\nu(\rho_B^{\otimes n})
\nu(\widehat\sigma_{B^n}).
\end{equation}
Applying Eq.~\eqref{eq:tensor-spectrum} (cf.~\Cref{lem:polynomial-spectrum}) to
$\tau_A^{\otimes n}$ and $\rho_B^{\otimes n}$, and applying
Eq.~\eqref{eq:poly-spectrum} (cf.~\Cref{lem:polynomial-spectrum}) to the permutation-invariant operator
$\widehat\sigma_{B^n}$, yields
\begin{equation}
v_n
\leq (n+1)^{d_A+d_B-2}
\left[
(n+1)^{d_B-1}
(n+d_B)^{d_B(d_B-1)/2}
\right],
\end{equation}
which proves Eq.~\eqref{eq:joint-block-bound}. In particular, $v_n$ grows
at most polynomially in $n$, and hence
\begin{equation}
\log v_n=O(\log n).
\end{equation}

We next establish the marginal identity Eq.~\eqref{eq:marginal-preserved}.
For any operator $X_{A^nB^n}$, a trace-preserving map acting only on $A^n$ does not affect the partial trace over $A^n$.
Thus,
\begin{align}
{\rm Tr}_{A^n}\!\left[
\mathcal E_n(\rho_{AB}^{\otimes n})
\right]
&={\rm Tr}_{A^n}\!\left[\cP_{\widehat\sigma_{B^n}}^{B^n}
 \circ\cP_{\rho_B^{\otimes n}}^{B^n}(\rho_{AB}^{\otimes n})\right]\\[1mm]
 &=\cP_{\widehat\sigma_{B^n}}^{B^n}
 \circ\cP_{\rho_B^{\otimes n}}^{B^n}(\rho_B^{\otimes n})\\[1mm]
 &=\cP_{\widehat\sigma_{B^n}}^{B^n}(\rho_B^{\otimes n})
 =\rho_B^{\otimes n},
\end{align}
where the second equality follows from the fact that the two pinchings acting on $B^n$ can be applied directly to the $B^n$-marginal, the third equality uses that $\cP_{\rho_B^{\otimes n}}(\rho_B^{\otimes n})=\rho_B^{\otimes n}$, and the last equality follows from the fact that $\cP_{\widehat\sigma_{B^n}}^{B^n}$ also leaves $\rho_B^{\otimes n}$ invariant because $\rho_B^{\otimes n}$ commutes with $\widehat\sigma_{B^n}$.

Because $\widehat\rho_{A^nB^n}$ is the output of the joint pinching, it is block diagonal with respect to the common refinement of the spectral decompositions.
It therefore commutes with
$\tau_A^{\otimes n}\otimes I_{B^n}$,
$I_{A^n}\otimes\rho_B^{\otimes n}$, and
$I_{A^n}\otimes\widehat\sigma_{B^n}$.
Consequently,
\begin{align}
[\widehat\rho_{A^nB^n},
\tau_A^{\otimes n}\otimes\rho_B^{\otimes n}]
&=0,\\[1mm]
[\widehat\rho_{A^nB^n},
\tau_A^{\otimes n}\otimes\widehat\sigma_{B^n}]
&=0,
\end{align}
which proves
Eqs.~\eqref{eq:commute-marginal-reference} and
\eqref{eq:commute-comparator-reference}.

It remains to verify the action of $\mathcal E_n$ on the two reference
operators. By spectral pinching,
\begin{equation}
\cP_{\rho_B^{\otimes n}}(\rho_B^{\otimes n})
=\rho_B^{\otimes n},\;\;
\cP_{\rho_B^{\otimes n}}(\sigma_B^{\otimes n})
=\widehat\sigma_{B^n}.
\end{equation}
Furthermore, the commutation relation $[\rho_B^{\otimes n},\widehat\sigma_{B^n}]=0$ implies that each of these two operators is invariant under the spectral pinching with respect to the other operator.
 Finally,
\begin{equation}
\cP_{\tau_A^{\otimes n}}(\tau_A^{\otimes n})
=\tau_A^{\otimes n}.
\end{equation}
Applying the three commuting pinchings in
Eq.~\eqref{eq:joint-pinching} therefore gives
\begin{align}
\mathcal E_n
(\tau_A^{\otimes n}\otimes\rho_B^{\otimes n})
&=\tau_A^{\otimes n}\otimes\rho_B^{\otimes n},\\[1mm]
\mathcal E_n
(\tau_A^{\otimes n}\otimes\sigma_B^{\otimes n})
&=\tau_A^{\otimes n}\otimes\widehat\sigma_{B^n}.
\end{align}
These are precisely
Eqs.~\eqref{eq:fix-marginal-reference} and
\eqref{eq:map-comparator-reference}, completing the proof.
\end{proof}

With \Cref{lem:joint-pinching} at our disposal, we are now ready to establish the desired result for the sandwiched R{\'e}nyi divergence when $\alpha\in[\frac{1}{2},1)$.

\begin{theorem}\label{thm:sandwiched1}
Let $\alpha\in[\frac12,1)$, let $\rho_{AB}$ and $\tau_A$ be  finite-dimensional states. Then
\begin{equation}
\tilde{D}_\alpha(\rho_{AB}\| \tau_A\otimes\rho_B) 
\leq \frac{1}{\alpha} \min_{\sigma_B\in\cS(B)}
\tilde{D}_\alpha(\rho_{AB}\| \tau_A\otimes\sigma_B).
\end{equation}
\end{theorem}
\begin{proof}
Fix an arbitrary state $\sigma_B\in\cS(B)$. If
\begin{equation}
\tilde{D}_\alpha
(\rho_{AB}\|\tau_A\otimes\sigma_B)=+\infty,
\end{equation}
then the desired inequality for this choice of $\sigma_B$ holds
trivially. We may thus assume that this divergence is finite.

For each $n$, we retain the notation introduced in
\Cref{lem:joint-pinching}. By the additivity of the sandwiched R{\'e}nyi divergence,
\begin{align}
n\tilde{D}_\alpha
(\rho_{AB}\|\tau_A\otimes\rho_B)
&=\tilde{D}_\alpha
\bigl(\rho_{AB}^{\otimes n}\big\|
\tau_A^{\otimes n}\otimes\rho_B^{\otimes n}\bigr)
\\[1mm]
&\leq\tilde{D}_\alpha
\bigl(\widehat\rho_{A^nB^n}\big\|
\tau_A^{\otimes n}\otimes\rho_B^{\otimes n}\bigr)
+\log v_n,
\label{eq:sandwiched-chain-one}
\end{align}
where the inequality follows from \cref{lem:pinching-loss} and
Eq.~\eqref{eq:fix-marginal-reference}.

By~Eq.~\eqref{eq:commute-marginal-reference},
$\widehat\rho_{A^nB^n}$ commutes with
$\tau_A^{\otimes n}\otimes\rho_B^{\otimes n}$. Consequently,
Eq.~\eqref{eq:commuting-equality} allows us to identify the corresponding
sandwiched R{\'e}nyi divergence with the Petz R{\'e}nyi divergence.
We therefore obtain
\begin{align}
\tilde{D}_\alpha
\bigl(\widehat\rho_{A^nB^n}\big\|
\tau_A^{\otimes n}\otimes\rho_B^{\otimes n}\bigr)
&=
D_\alpha
\bigl(\widehat\rho_{A^nB^n}\big\|
\tau_A^{\otimes n}\otimes\rho_B^{\otimes n}\bigr)
\\
&\leq\frac{1}{\alpha}
\min_{\sigma_{B^n}}
D_\alpha
\bigl(\widehat\rho_{A^nB^n}\big\|
\tau_A^{\otimes n}\otimes\sigma_{B^n}\bigr)
\\
&\leq\frac{1}{\alpha}
D_\alpha
\bigl(\widehat\rho_{A^nB^n}\big\|
\tau_A^{\otimes n}\otimes\widehat\sigma_{B^n}\bigr)
\\
&=\frac{1}{\alpha}
\tilde{D}_\alpha
\bigl(\widehat\rho_{A^nB^n}\big\|
\tau_A^{\otimes n}\otimes\widehat\sigma_{B^n}\bigr).
\label{eq:sandwiched-chain-two}
\end{align}
Here, the first inequality follows by applying
\cref{thm:petz} to the bipartition $A^n:B^n$, noting that the
$B^n$-marginal of $\widehat\rho_{A^nB^n}$ is
$\rho_B^{\otimes n}$. The final equality follows from
Eq.~\eqref{eq:commute-comparator-reference}.

Applying the data-processing inequality for the sandwiched R{\'e}nyi
divergence to the channel $\mathcal E_n$ and using
Eq.~\eqref{eq:map-comparator-reference}, we further obtain
\begin{align}
\tilde{D}_\alpha
\bigl(\widehat\rho_{A^nB^n}\big\|
\tau_A^{\otimes n}\otimes\widehat\sigma_{B^n}\bigr)
&\leq
\tilde{D}_\alpha
\bigl(\rho_{AB}^{\otimes n}\big\|
\tau_A^{\otimes n}\otimes\sigma_B^{\otimes n}\bigr)
\\[1mm]
&=n\tilde{D}_\alpha
(\rho_{AB}\|\tau_A\otimes\sigma_B),
\label{eq:sandwiched-chain-three}
\end{align}
where the equality again follows from additivity.

Combining Eqs.~\eqref{eq:sandwiched-chain-one}--
\eqref{eq:sandwiched-chain-three} and dividing by $n$ yields the
finite-blocklength estimate
\begin{equation}
\label{eq:sandwiched-finite-block}
\tilde{D}_\alpha(\rho_{AB}\|\tau_A\otimes\rho_B)
\leq\frac{1}{\alpha}
\tilde{D}_\alpha(\rho_{AB}\|\tau_A\otimes\sigma_B)
+\frac{\log v_n}{n}.
\end{equation}
By~Eq.~\eqref{eq:joint-block-bound},
\begin{equation}
\lim_{n\to\infty}\frac{\log v_n}{n}=0.
\end{equation}
Taking $n\to\infty$ in Eq.~\eqref{eq:sandwiched-finite-block} therefore
gives
\begin{equation}
\tilde{D}_\alpha(\rho_{AB}\|\tau_A\otimes\rho_B)
\leq\frac{1}{\alpha}
\tilde{D}_\alpha(\rho_{AB}\|\tau_A\otimes\sigma_B).
\end{equation}
Since $\sigma_B\in\cS(B)$ was arbitrary, taking the infimum over
$\sigma_B$ proves the asserted inequality and completes the proof.
\end{proof}

\subsection{Sandwiched--R\'enyi bound with Order \texorpdfstring{$\alpha\in(0,\frac{1}{2})$}{α in (0,1/2)}}
We now prove the desired result for the sandwiched R{\'e}nyi divergence in the range $\alpha\in(0,\frac{1}{2})$. The proof relies on the following auxiliary lemmas.

\begin{lemma}\label{lem:Tr-inequality}
Let $X\geq0$ and $s\in(0,1)$, then we have for any pinching map $P[\cdot]$ with $k$ projections:
\begin{equation}
\tr[X^s]\leq\tr[P(X)^s]\leq k^{1-s}\tr[X^s]
\end{equation} 
\end{lemma}
\begin{proof}
The first inequality follows from Jensen's Operator inequality~\cite{HANSEN_2003} and $x\mapsto x^s$ being operator concave. For the second inequality, we first assume that $X>0$. By the pinching inequality $X\leq k P(X)$, both operators are positive definite and taking inverses reverses the order $X^{-1} \geq (kP(X))^{-1}$. Then using that $x\mapsto x^{1-s}$ is operator monotone, we get 
\begin{equation}
\tr[X^s] = \tr[X X^{s-1}] \geq k^{s-1} \, \tr[X(P(X))^{s-1}] = k^{s-1} \, \tr[P(X)^s] \, .
\end{equation}
For a general $X\geq0$, apply the above argument to
$X_\varepsilon=X+\varepsilon I>0$. Since
$P(X_\varepsilon)=P(X)+\varepsilon I$, letting
$\varepsilon \rightarrow 0$ and using the continuity of $A\mapsto\tr[A^s]$
gives the desired inequality.
\end{proof}
\begin{lemma}[Theorem~4.4 of \cite{CARLEN2016174}]\label{lem:Tr-concave}
Let $X,Y\geq0$ and $\alpha \in(0,1/2)$, then
\begin{equation}
f_\alpha(X,Y) \coloneqq \tr[(Y^{1/2}XY^{1/2})^\alpha]
\end{equation}
is jointly concave and for any pinching map $f_\alpha(P(X),P(Y)) \geq f_\alpha(X,Y)$.
\end{lemma}
Using \Cref{lem:Tr-inequality} and \Cref{lem:Tr-concave}, we establish the following lemma, which is key to deriving the desired result for $\alpha\in(0,\frac{1}{2})$.
\begin{lemma}\label{lem:almostdpi}
Let $\rho,\sigma$ be states, $\alpha \in(0,1/2)$ and $\delta = \frac{1-\alpha}{\alpha}>1$. Then for any pinching map $P$ with $v$ projections, let $\eta = \frac{{P(\sigma^\delta)}^{\frac{1}{\delta}}}{z}$ with $z= \tr\left[{P(\sigma^\delta)}^\frac{1}{\delta}\right]$.
Then we have
\begin{equation}
\Tilde{D}_\alpha ( P(\rho) \| \eta) \leq \Tilde{D}_\alpha (\rho \| \sigma) + \log v
\end{equation}
\end{lemma}
\begin{proof}
Using the definition of the sandwiched R\'enyi divergence, \Cref{lem:Tr-concave}, and the relation
$P(\sigma^\delta)=z^\delta\eta^\delta$, while noting that $\frac{1}{\alpha-1}<0$, we obtain 
\begin{align}
\Tilde{D}_\alpha (\rho \| \sigma)
&= \frac{1}{\alpha-1}\log f_{\alpha}(\rho , \sigma^{\delta})\\
\overset{\mathrm{\Cref{lem:Tr-concave}}}&{\geq} \frac{1}{\alpha-1}\log f_{\alpha}(P(\rho) , P(\sigma^{\delta}))\\
&= \frac{1}{\alpha-1}\log f_{\alpha}(P(\rho) , z^{\delta}\eta^{\delta})\\
&=\frac{1}{\alpha-1}\log f_{\alpha}(P(\rho) , \eta^{\delta}) - \log z\\[1mm]
&= \Tilde{D}_\alpha ( P(\rho) \| \eta) - \log z.
\end{align}
For the normalization factor $z$, applying \Cref{lem:Tr-inequality} to $X=\sigma^\delta$ with $s=1/\delta\in(0,1)$ gives 
\begin{equation}
z= \tr\left[{P(\sigma^\delta)}^\frac{1}{\delta}\right]\overset{\mathrm{\Cref{lem:Tr-inequality}}}{\leq}
v^{1-s}\tr[\sigma]=v^{1-s} \leq v.
\end{equation}
This completes the proof.
\end{proof}
With the above lemma in hand, we can adapt the proof of \Cref{thm:sandwiched1} to extend its conclusion to the range $\alpha\in(0,\frac{1}{2})$.
\begin{theorem}\label{thm:sandwiched2}
Let $\alpha\in(0,1/2)$, let $\rho_{AB}$ and $\tau_A$ be  finite-dimensional states. Then
\begin{equation}
\tilde{D}_\alpha(\rho_{AB}\| \tau_A\otimes\rho_B) 
\leq \frac{1}{\alpha} \min_{\sigma_B\in\cS(B)}
\tilde{D}_\alpha(\rho_{AB}\| \tau_A\otimes\sigma_B).
\end{equation}
\end{theorem}
\begin{proof}
Fix an arbitrary state $\sigma_B\in\cS(B)$. If
\begin{equation}
\tilde{D}_\alpha(\rho_{AB}\|\tau_A\otimes\sigma_B)=+\infty,
\end{equation}
then the desired inequality for this choice of $\sigma_B$ holds
trivially. We therefore assume that the divergence is finite.
For each $n>1$, 
we introduce
\begin{equation}
    \eta_n = \frac{{\mathcal P_{\rho_B^{\otimes n}}\left((\sigma_B^{\otimes n})^\delta\right)}^{\frac{1}{\delta}}}{z}, 
\end{equation}
where $\delta = \frac{1-\alpha}{\alpha}$ and $z= \tr[{\mathcal P_{\rho_B^{\otimes n}}\left((\sigma_B^{\otimes n})^\delta\right)}^{\frac{1}{\delta}}]$.
Then define
\begin{equation}
    \cF_n\coloneqq \cP_{\tau_A^{\otimes n}}^{A^n} \circ \cP_{\eta_n}^{B^n}\circ\cP_{\rho_B^{\otimes n}}^{B^n}.
\end{equation}
Since $[\eta_n,\rho_B^{\otimes n}]=0$, their pinchings commute. Hence $\cF_n$ is a pinching by their non-zero products. Analogously to~\cref{lem:joint-pinching}, we can bound the number of projections by
\begin{equation}
    v_n \leq \nu(\tau_A^{\otimes n})\nu(\rho_B^{\otimes n})\nu(\eta_n ) \leq \poly(n) ,
\end{equation}
and have 
\begin{align}
 \tr_{A^n}  \cF_n(\rho_{AB}^{\otimes n}) &= \rho_B^{\otimes n}\\[1mm]
    [\cF_n(\rho_{AB}^{\otimes n}),\tau_A^{\otimes n} \otimes \rho_B^{\otimes n}]&=0, \\ [\cF_n(\rho_{AB}^{\otimes n}),\tau_A^{\otimes n} \otimes \eta_n] &=0 , \\[1mm]
    \cF_n(\tau_A^{\otimes n} \otimes \rho_B^{\otimes n}) &= \tau_A^{\otimes n} \otimes \rho_B^{\otimes n}, \\[1mm]
    \cF_n( (\tau_A^{\otimes n} \otimes \sigma_B^{\otimes n})^\delta) &= z^\delta (\tau_A^{\otimes n} \otimes \eta_n)^\delta, 
\end{align}

Adapting the argument used in the proof of \Cref{thm:sandwiched1}, we obtain 
\begin{align}
n \tilde{D}_\alpha(\rho_{AB}\| \tau_A \otimes \rho_B) &= \tilde{D}_\alpha(\rho_{AB}^{\otimes n}\| \tau_A^{\otimes n} \otimes \rho_B^{\otimes n}) \\[1mm]
&\leq \tilde{D}_\alpha(\cF_n(\rho_{AB}^{\otimes n})\| \tau_A^{\otimes n} \otimes \rho_B^{\otimes n}) + \log v_n \\[1mm]
 &= D_\alpha(\cF_n(\rho_{AB}^{\otimes n})\| \tau_A^{\otimes n} \otimes \rho_B^{\otimes n}) + \log v_n \\ 
 &\leq \frac{1}{\alpha} \min_{\sigma_{B_n}}D_\alpha(\cF_n(\rho_{AB}^{\otimes n})\| \tau_A^{\otimes n} \otimes \sigma_{B_n}) + \log v_n \\ 
 &\leq \frac{1}{\alpha} D_\alpha(\cF_n(\rho_{AB}^{\otimes n})\| \tau_A^{\otimes n} \otimes \eta_n) + \log v_n \\ 
 &= \frac{1}{\alpha} \tilde{D}_\alpha(\cF_n(\rho_{AB}^{\otimes n})\| \tau_A^{\otimes n} \otimes \eta_n) + \log v_n \\
 &\leq \frac{1}{\alpha} \tilde{D}_\alpha(\rho_{AB}^{\otimes n}\| \tau_A^{\otimes n} \otimes \sigma_B^{\otimes n}) + \left(1+\frac{1}{\alpha}\right)\log v_n \\
 &= n \frac{1}{\alpha} \tilde{D}_\alpha(\rho_{AB}\|\tau_A \otimes \sigma_B) +  \left(1+\frac{1}{\alpha}\right)\log v_n.
\end{align}
Here, the first and final equalities follow from the additivity of the sandwiched R\'enyi divergence. The first inequality follows from~\Cref{lem:pinching-loss}. The second equality holds because $\cF_n(\rho_{AB}^{\otimes n})$ commutes with $\tau_A^{\otimes n}\otimes\rho_B^{\otimes n}$, while the second inequality follows from~\Cref{thm:petz}. The third equality holds because $\cF_n(\rho_{AB}^{\otimes n})$ commutes with $\tau_A^{\otimes n} \otimes \eta_n$, and the final inequality follows because all requirements for~\Cref{lem:almostdpi} are fulfilled. Consequently, for every $\sigma_B\in\cS(B)$,
\begin{equation}
\tilde{D}_\alpha(\rho_{AB}||\tau_A\otimes\rho_B) \leq \frac{1}{\alpha} \tilde{D}_\alpha(\rho_{AB}||\tau_A\otimes\sigma_B) + \frac{\left(1+\frac{1}{\alpha}\right)}{n}\log\poly(n).
\end{equation}
For every fixed $a\in(0,1/2)$, taking $n \to \infty$ and then the infimum over $\sigma_B$ proves the theorem.
\end{proof}

\begin{remark}[Tightness of the sandwiched--R\'enyi bound]
For $\alpha\in[1/2,1)$, the example in Remark~\ref{rem:tight} also shows that
the corresponding bound for the sandwiched R{\'e}nyi divergence is
tight. Indeed, regarding $P_{XY}$ as a classical--classical state and
applying the data-processing inequality to the dephasing channel on
$Y$, the minimization over all states $\sigma_Y$ can be restricted to
states that are diagonal in the classical basis. Since the Petz and
sandwiched R{\'e}nyi divergences coincide for commuting states, the
sandwiched quantities reduce exactly to the corresponding Petz
quantities considered in Remark~\ref{rem:tight}. The claimed tightness
therefore follows directly from Remark~\ref{rem:tight}.

For $\alpha\in(0,1/2)$, the example in Remark~\ref{rem:tight} also shows that
the sandwiched R{\'e}nyi marginal bound is tight, although the
data-processing argument used above is no longer available.
Let $\rho_{XY}$ be the classical--classical state associated with
$P_{XY}$. The sandwiched marginal bound gives
\begin{equation}
0
\leq
\frac{1}{\alpha}
\min_{\sigma_Y\in\mathcal S(Y)}
\tilde{D}_\alpha
\bigl(P_{XY}\big\|U_X\otimes\sigma_Y\bigr)
-
\tilde{D}_\alpha
\bigl(P_{XY}\big\|U_X\otimes\rho_Y\bigr).
\end{equation}
On the other hand, the diagonal states form a subset of
$\mathcal S(Y)$. Hence,
\begin{align}
&\frac{1}{\alpha}
\min_{\sigma_Y\in\mathcal S(Y)}
\tilde{D}_\alpha
\bigl(P_{XY}\big\|U_X\otimes\sigma_Y\bigr)
-
\tilde{D}_\alpha
\bigl(P_{XY}\big\|U_X\otimes\rho_Y\bigr)
\\
&\leq
\frac{1}{\alpha}
\min_{R_Y}
\tilde{D}_\alpha
\bigl(P_{XY}\big\|U_X\times R_Y\bigr)
-
\tilde{D}_\alpha
\bigl(P_{XY}\big\|U_X\times P_Y\bigr) \\
&=
\frac{1}{\alpha}
\min_{R_Y}
D_\alpha
\bigl(P_{XY}\big\|U_X\times R_Y\bigr)
-
D_\alpha
	\bigl(P_{XY}\big\|U_X\times P_Y\bigr),
\end{align}
where we used the fact that the Petz and sandwiched R{\'e}nyi
divergences coincide for commuting states. The right-hand side
converges to zero by Remark~\ref{rem:tight}. Therefore,
\begin{equation}
\lim_{n\to\infty}
\left[
\frac{1}{\alpha}
\min_{\sigma_Y\in\mathcal S(Y)}
\tilde{D}_\alpha
\bigl(\rho_{XY}\big\|U_X\otimes\sigma_Y\bigr)
-
\tilde{D}_\alpha
\bigl(\rho_{XY}\big\|U_X\otimes\rho_Y\bigr)
\right]
=0.
\end{equation}
Thus, the sandwiched R{\'e}nyi marginal bound is also tight for every
$\alpha\in(0,1/2)$.
\end{remark}

\section{AI Statement}
The authors thank ChatGPT with GPT-5.6 Sol for helping us identify the key pinching-map construction used in~\Cref{lem:joint-pinching}, which plays a crucial role in establishing the result for the sandwiched Rényi divergence when $\alpha\in (0,1)$ and further help with~\Cref{lem:almostdpi}.
All mathematical statements
and proofs were independently verified by the authors. The authors take responsibility for the final manuscript.
\section{Acknowledgments} 
We acknowledge support from the Quantum Center at ETH Zurich. This work has received funding from the Swiss State Secretariat for Education, Research and Innovation (SERI) under contract No.\ 20QU-1\_225224. The research of Gui-long Jiang and Shi-Bing Li were supported by the National
Natural Science Foundation of China (Grant No. 62571166).

\appendix

\section{Non-full-rank states}
\label{app:1}
The full-rank assumption on $\tau_A$ can be removed by continuity. Set \begin{equation} \tau_A^\varepsilon=(1-\varepsilon)\tau_A +\varepsilon\frac{\id_A}{d_A}, \qquad \varepsilon\in(0,1), \end{equation} and define \begin{equation} Q_\varepsilon(\sigma_B) := \operatorname{tr}\!\left[ \rho_{AB}^{\alpha} \bigl(\tau_A^\varepsilon\otimes\sigma_B\bigr)^{1-\alpha} \right]. \end{equation} Since $x\mapsto x^{1-\alpha}$ is continuous on the positive semidefinite cone, \begin{equation} \bigl\|(\tau_A^\varepsilon)^{1-\alpha} -\tau_A^{1-\alpha}\bigr\|_\infty\longrightarrow 0 . \end{equation} Moreover, $\|\sigma_B^{1-\alpha}\|_\infty\leq 1$ for every state $\sigma_B$. Hence $Q_\varepsilon\to Q_0$ uniformly on $\mathcal{S}(B)$, and therefore \begin{equation} \max_{\sigma_B}Q_\varepsilon(\sigma_B) \longrightarrow \max_{\sigma_B}Q_0(\sigma_B). \end{equation} Because $(\alpha-1)^{-1}<0$, minimizing the Petz--R\'enyi divergence is equivalent to maximizing $Q_\varepsilon$. Taking $\varepsilon\to0$ in the full-rank result consequently proves \begin{equation} D_\alpha(\rho_{AB}\|\tau_A\otimes\rho_B) \leq \frac{1}{\alpha} \min_{\sigma_B} D_\alpha(\rho_{AB}\|\tau_A\otimes\sigma_B) \end{equation} for arbitrary states $\tau_A\geq0$.
\bibliographystyle{arxiv_no_month}
\bibliography{bibliofile}

\end{document}